\documentclass[10pt,journal]{IEEEtran}
\usepackage{multirow}
\usepackage{amssymb}
\usepackage{graphicx}
\usepackage{verbatim}
\usepackage{cite}
\usepackage{subfig}
\usepackage{amsthm}
\usepackage{caption}
\newtheorem{theorem}{Theorem}
\newtheorem{definition}{Definition}

\usepackage{ifpdf}

\usepackage{cite}

\ifCLASSINFOpdf
\else
\fi
\usepackage[cmex10]{amsmath}
\usepackage{amsfonts}

\usepackage{algorithmic}
\usepackage{algorithm}

\begin{document}
\title{Privacy-preserving Network Functionality Outsourcing}

\author{
\IEEEauthorblockN{\emph{\today}\\
 Junjie Shi, Yuan Zhang and Sheng Zhong}

\IEEEauthorblockA{
%Deparment of Computer Science and Technology\\
National State Key Laboratory for Novel Software Technology\\
Nanjing University, Nanjing 210023, China\\
junjie.shi@smail.nju.edu.cn, zhangyuan@nju.edu.cn and zhongsheng@nju.edu.cn}
}



\maketitle


\begin{abstract}
Since the advent of software defined networks (\mbox{SDN}), there have been many attempts  to outsource the complex and costly local network functionality, i.e. the middlebox, to the cloud in the same way as outsourcing computation and storage. The privacy issues, however, may thwart the enterprises' willingness to adopt this innovation since the underlying configurations of these middleboxes may leak crucial and confidential information which can be utilized by attackers. To address this new problem, we use firewall as an sample functionality and propose the first  privacy preserving outsourcing framework and  schemes in SDN. The basic technique that we exploit is a ground-breaking tool in cryptography, the \textit{cryptographic multilinear map}. In contrast to the infeasibility in efficiency if a naive approach is adopted, we devise practical schemes that can outsource the middlebox as a blackbox after \textit{obfuscating} it such that the cloud provider can efficiently perform the same functionality without knowing its underlying  private configurations. Both theoretical analysis and experiments on real-world firewall rules demonstrate that our schemes are secure, accurate, and practical.
%Besides, we emphasize that our framework and schemes can also be modified to support some other packet header analysis based functionalities, not limited to firewalls.
\end{abstract}


%\begin{IEEEkeywords}

%\end{IEEEkeywords}

% \fi
\IEEEpeerreviewmaketitle

\section{Introduction}

Although network functionalities play an important role in enterprise network to make it robust, fast, and secure, they often burden a company with great hardware financial pressure and management complexity. Network middleboxes, such as firewalls  and intrusion detection systems (IDS), are the customized appliances to implement these  sophisticated functionalities, and they are also often hard to deploy and upgrade. A recent survey \cite{sherry2012making} reveals that the investments in network infrastructures deployment of these middleboxes as well as in the personnel cost of managing and maintaining them are substantial.
%
%Thus it is desirable for us to outsource them to cloud provider, in the same way that the computation and storage services have been outsourced.

%The networks of a company often consist of many network devices and encompass a great deal of software to control the basic traffic flows and enable various network functionalities.  For example, the firewall system, a typical middlebox, is installed in almost  every enterprise to protect their networks. These middleboxes are vital for making an enterprise network robust, fast, and secure. But,

The emergence of SDN helps us separate the logic control from the basic traffic processing, so we can  relieve some of the above ``pain points'' to a certain extent by taking the advantage of SDN\cite{casado2007ethane,gember2012toward,qazi2013simple,DPIaS14}.
%For example,
%the middlebox placement problem can be solved by redirecting the related traffic through a position wherever the middlebox is actually placed, all under the management of the control plane in SDN\cite{casado2007ethane}. Besides, some new network functionalities or policies can be ``inserted'' flexibly into the network  by specifying the functions or rules in the central logic control plane \cite{casado2007ethane}.
%SIMPLE\cite{qazi2013simple} has been proposed as a SDN-based policy enforcement layer to ``translate'' logical middlebox policies to the traffic configurations of the data plane. This new layer will enable more efficient middleboxes composition, load balance and packet modification with the power of SDN.
But though the management of networks can be simplified by SDN, it still has many shortcomings. For instance, because these functionalities are still implemented within the enterprise network, the hardware cost and the everyday maintenance cost are still there.

Thus, to further reduce the cost and complexity faced by local networks, some SDN researchers have attempted to outsource these network functionalities or middleboxes out to cloud providers, in the same way that the computation and storage services have been successfully outsourced \cite{sherry2012making, gibb2012outsourcing}. After migrating  middleboxes to cloud, the local enterprises will receive better services with less expenditure and management complexity since the cloud provider can deploy more advanced hardware and hire professional experts to offer better network functionality services by taking advantage of economies of scale.  Moreover, the local enterprise can also have a clearer and more elastic network infrastructure.

However, the outsourcing of local SDN middleboxes  may bring in serious privacy issues\cite{sherry2012making}. Because at present, an enterprise has to provide the detailed underlying configurations of these middleboxes, which may embody very sensitive and important enterprise secrets, when outsourcing middleboxes. For example, the configurations of firewalls or IDS may reveal what the enterprise network topology looks like and how different sectors within an enterprise are regulated. These policies or configurations are set and  known by a restricted few administrators even before outsourcing. Were these confidential pieces of information hidden in these polices to be exposed to adversaries, they may be utilized to impose a great threat to the enterprise\cite{SRDS_Liu2012}.
%So the security and privacy issues can be  major obstacles for the adoption of SDN network middlebox outsourcing.
%The privacy challenge in this scenario can be more severe than that in outsourcing storage or computation since these outsourced configurations have to be subtly used by cloud providers to enable these functionalities.

In brief, the new challenges in this scenario now we face are: 1) we need to  protect the private information hidden in the configurations of these  outsourced middleboxes; 2) we need to allow the cloud provider to perform the network functionalities without mistakes ( or with acceptable errors). Note that in line with existing outsourcing architecture in SDN\cite{sherry2012making,karaoglu2013offloading,gibb2012outsourcing}, we  assume that the cloud provider for an enterprise is its Internet Service Providers (ISP), so we will not consider the protection of the packet flows of the enterprise here.

In this paper, we try to address these challenges by designing a privacy preserving scheme for network middlebox outsourcing in SDN, exploiting a \textit{cryptographic multilinear map}, a powerful cryptographic tool that was recently invented \cite{zvika2013obfus,coron2013practical,garg2013candidate}. To keep things concrete, in this paper we will take firewalls, a typical and indispensable  network middlebox, as a case study.

The main idea of our scheme is to use cryptographic multilinear map to obfuscate the rules in a firewall efficiently such that the original configurations are protected while their functionality is preserved. The theoretical foundation of our idea is cryptographic \textit{program obfuscation} \cite{barak2012possibility,programobf}.
%But
%In our solution, we will first model this problem as how to construct a cryptographic obfuscator for the firewall.
%A cryptographic obfuscator of a program can make a program perform the same functionality while the detailed underlying instructions are hidden. The last two years have seen great advances in this important cryptography field, along with the development of the multiliear map.
In our work, we  first present framework SOFA (\underline{S}ecure framework for \underline{O}utsourcing \underline{F}irew\underline{A}ll). This framework consists of two phases. First, the controller in the enterprise SDN constructs a cryptographic obfuscator for the firewall based on its rules. This obfuscator encodes  the firewall rules such that the underlying confidential information cannot be derived. Then this enterprise transmits this obfuscated firewall to a  provider.
%, which consists of encodings and some public parameters,
In the second phase, the cloud provider  performs the firewall functionality by filtering incoming and outgoing network traffics with this obfuscated firewall without knowing the detailed configurations.
%This privacy preserving framework can be compatible with existing SDN outsourcing architecture, such as APLOMB\cite{sherry2012making}.
Based on this overall framework, a basic scheme and two enhanced schemes are proposed to address the detailed technical issues.

To summarize, we have made three main contributions in our work:
\begin{itemize}
\item To the best of our knowledge, we are the first to deal with the privacy problem about  network functionality outsourcing (NFO) in SDN. The obfuscation approach we propose has a solid theoretical foundation. And framework SOFA is compatible with existing SDN outsourcing architecture, such as APLOMB\cite{sherry2012making}.
%and can be used We model the secure network middlebox outsourcing as a kind of program obfuscation problem such that its security has a solid foundation. And we propose a basic privacy preserving outsourcing scheme based on a latest breakthrough in cryptography, to protect the outsourced firewall rules. No interaction between the enterprises and the cloud provider is needed once the obfuscated firewall is outsourced.
\item To outsource firewalls as a concrete example, we propose a basic scheme which is efficient in the first phase, and two enhanced schemes which can further improve the efficiency in the second phase.
%The security of these schemes is analyzed. of also put up some approaches to improve our basic scheme's efficiency while the sacrifices on privacy are acceptable such that our algorithms can be more applicable in practice.
We also analyse the complexity and security  of our algorithms to show that they are efficient, secure, and reliable.
\item We implement a prototype of our schemes, and evaluate its performance on real-world firewall rules, the  results show that our proposed schemes are practical.
\end{itemize}

Our paper are organized as follows. First, we review some related work in section \ref{sec2}. In section \ref{sec3}, we give formal definition to the privacy issues in network functionality outsourcing and introduce some basic cryptographic backgrounds. In section \ref{sec4}, we present framework SOFA and a basic scheme. Some approaches which can boost its efficiency are discussed in \ref{Improve}. The security of these algorithms is analysed in section \ref{sec5}. Evaluation of SOFA with firewall rules is presented in section \ref{sec6}.  Finally, in section \ref{sec7} we conclude our work and discuss some potential future research directions.

\section{Related Works}
\label{sec2}

%Some recent works on networks, especial with the advent of SDN, have initiated the process towards outsourcing local network middleboxes to clouds.
In this section, we review some recent works on outsourcing frameworks in SDN, some privacy-preserving outsourcing techniques in cloud computing or traditional enterprise networks.
%tasks in the cloud and some representative privacy-preserving techniques that have been used in outsourcing data storage and computation.
We can see that no existing work can effectively address the novel challenges we have proposed.

\subsection{Outsourcing network functionality in SDN}
There has already been much research done about the outsourcing of network middleboxes services exploiting the advantage of SDN\cite{gember2013stratos,sherry2012making,karaoglu2013offloading,gibb2012outsourcing,kotronis2012outsourcing}.
Sherry et al.\cite{sherry2012making} design and implement a prototype of APLOMB architecture, by which the middlebox processing within enterprises can be outsourced to clouds. This architecture is mainly intended to deal with the problem of how to redirect the traffic to/from the enterprise SDN to the cloud without greatly damaging performance. It is compatible with many specific middleboxes, such as firewalls.
%In addition, this work admits that APLOMB faces security challenges, which may affect the companies' willingness to adopt this architecture.
%Our work in this paper may be added with ease to this architecture to make the configurations of middleboxes secure.
Gibb et al.\cite{gibb2012outsourcing} also envision outsourcing these network functionalities to external \textit{Feature Providers}. The main goal of their work is to provide anyone with chances to be a network service provider without any limitation on location.
%The prototype Jingling that they design provides an interface, Feature API, through which an enterprise SDN and \textit{Feature Providers} can communicate with each other to negotiate about the policy information and configuration details.
\textit{Cloud-Assisted Routing}\cite{karaoglu2013offloading} is another scheme aimed at offloading the local computation-intensive and memory-intensive operations in the routing functionality to the cloud. %such that  the local management can be simplified and the routing scalability can be improved.

There is only a little current research concerning the privacy and security problem of outsourcing network functionality in SDN. Fayazbakhsh et al.\cite{fayazbakhsh2013verifiable}  mainly discuss the roadmap to construct a verifiable NFO architecture that can verify the correctness on intended functionality, performance assurance and can audit the actual workload in the cloud.

However, none of the these works about outsourcing in SDN can address the privacy challenges we have stated.

\subsection{Privacy Preserving outsourcing in Cloud Computing}

In literature, many schemes are designed to tackle the privacy problem of  outsourcing  data storage and computation  rather than the network middlebox\cite{wang2010privacy,wei2014security}. However, because the network functionality is often performed on real-time packet traffics in the cloud, its requirement is different. For example, the well-known public auditing framework \cite{wang2010privacy} introduces a third party auditor to help cloud users check the integrity of their outsourced data. However, in this framework the user files are sent to the cloud directly without encryption, so it will not be applicable to network middlebox outsourcing. So these existing approaches for cloud computing or storage are incapable of addressing this new problem.

Khakpour and Liu\cite{SRDS_Liu2012} present Landon, a framework to address the problem of how to protect the firewall policies when  it is outsourced to a cloud provider in traditional enterprise networks. The main idea of this framework is to  convert the raw access control lists to an equivalent Firewall Decision Diagram (FDD) and then to anonymize the FDD with Bloom Filters. This framework works well after some complex optimization methods are used. The limitation of this framework is that it is based on a specific presentation (FDD) of firewall and the Bloom Filter is not very secure and has some drawbacks\cite{rottenstreich2012bloom}.

\subsection{Related Cryptographical Research}
In cryptography, there is also some theoretical work dealing with similar problems, such as methods based on Yao's \emph{garbled circuits} for secure multi-party computation\cite{yao1986generate}. In particular, recently Brakerski and Rothblum\cite{zvika2013obfus} propose an elegant construction for obfuscating conjunctions based on an asymmetric multilinear map. This is the fundamental cryptographic tool we will use in this paper.
%and we will introduce it briefly in next section.
But their work focuses on the theoretical construction and security proof of obfuscating a single conjunction. So their result, though significant in theory, will be impractical in practice in directly obfuscating  numerous firewall rules. We will discuss this problem at length in section \ref{sec4}.
%but our work is efficient for usage exploiting the property of firewall in our algorithms.

\section{Problem formulation and Preliminaries}
\label{sec3}
In this section, we formulate the privacy problems of NFO in SDN. To this end, we first review the overall architecture for outsourcing network functionality, and give its  threat model. Specifically, we  give a definition for privacy preserving outsourcing firewalls.
In addition, we also recall some important cryptographic backgrounds about multilinear map.
%and its approximate implementation by the graded encoding system in this section.
%At last, we give a brief introduction to program obfuscation.

\subsection{System Model}
The overall architecture for outsourcing network functionality to cloud has been proposed, such as APLOMB\cite{sherry2012making} and Jingling\cite{gibb2012outsourcing}. In such architecture, enterprise administrators specify the policies of middleboxes with its local SDN controller in a simple way. Then by the communication mechanism between the enterprise clients and the cloud provider, these policies are transmitted to the cloud controller and service requirements are also negotiated. The controller in cloud SDN then translate these policies into low-level network configurations for middleboxes, and then manage the cloud network resources to carry out the tasks. The traffic redirection of the enterprise network flows is also managed by this architecture.
We will focus on the privacy protection of middlebox configurations in the above related procedures.

\subsection{Threat Model}
We assume that the cloud middlebox provider is \textit{semi-honest} since in the  outsourcing system model, the cloud providers are commonly the ISPs for enterprises. Because the provider performs  network tasks on packet traffics of the enterprise, it will be able to observe the input packets and the output actions of them. So the \textit{semi-honest} provider may be curious to learn information from these input-output histories. But the provider should not actively generate tremendous packets to attack the outsourced middlebox $\mathcal{MB'}$ for such an attack will affect the service quality it provides. So the two main threats we consider are the analysis of the outsourced $\mathcal{MB'}$ itself, and the analysis of its running input-output history.

\begin{table}[!t]
\renewcommand{\arraystretch}{1.3}
  \centering
  \caption{ Firewall Rule Examples  }\label{table1}
  \begin{tabular}{@{} c |c |c |c |c |c |c @{}}
  \hline
  \multirow{2}{0.5cm}{rule \#}  & \multicolumn{2}{c|}{Source}  & \multicolumn{2}{c|}{Destination} & \multirow{2}{0.8cm}{Protocol} & \multirow{2}{0.8cm}{Action} \\
    \cline{2-3} \cline{4-5}
    & IP & Port & IP & Port & & \\
   \hline
  1 &  192.168.45.* & * & * & * & * & deny \\
   \hline
  %\cline{2-3} \cline{4-5}
  2 & 10.*.*.* & * &192.168.4.* &80 &TCP &permit \\
   \hline
  %\cline{2-3} \cline{4-5}
  3 & 10.56.*.* &* &192.168.*.* &[22,88] &TCP &permit  \\
    \hline
  %\cline{2-3} \cline{4-5}
  4 & 114.212.190.* &8000 &* &8090 &UDP &deny  \\
   \hline
\end{tabular}
\end{table}
\subsection{Problem Definition}
Different middleboxes have their own characteristic in detailed configurations, thus in the following we focus on the firewall as an example.
The configuration of a typical IP firewall \textit{f} is an ordered sequence of rules, and they are usually presented in the format of access control lists (ACLs). Some examples are shown in Table \ref{table1}. Technically, the rules can be classified as two types: standard ACLs and extended ACLs. The standard rules only allow you to deny or permit packets by specifying the source IP addresses, like the first one in Table \ref{table1}. Extended rules also allow you to specify the protocol type, the destination address and the port ranges, like the other examples in table \ref{table1}. Once the configuration are settled, the firewall can perform packet filtering based on these rules to protect an enterprise network.

The rules of the firewall are what we need to protect since they may reveal, for example, much information of the enterprise's inner network topology and its security vulnerabilities. Specifically, we think that the IP address is more important than TCP port since the IP address determines the host end while the TCP ports are for specific processes in a host.  Formally, we give the following definition:

\begin{definition}(\textbf{Privacy Preserving Firewall Problem})
\label{PPFP}
The Privacy preserving of a firewall \textit{f} is the construction of a semantically equivalent secure firewall \textit{f\'} such that (1) for every packet \textit{p}, the probability of the decisions on it by \textit{f\'} and \textit{f} are different is negligible,  (2) it should be computationally efficient to perform the filtering tasks with  \textit{f\'}, and (3) there should be no computationally efficient ways to get the whole original plain rules of  \textit{f} based on  \textit{f\'} itself.
\end{definition}

It should be easy for us to give  privacy preserving definitions for other network middleboxes similarly.

\subsection{Preliminaries}
Here, we first review the definition of  \textit{ cryptographic multilinear map}, which is the cornerstone of our algorithms.

\begin{definition}(\textbf{$\kappa$-multilinear Map}\cite{garg2013candidate,boneh2003applications})
For $\kappa+1$ cyclic groups $G_1,\dots, G_\kappa, G_T$ of the same order $p$, a \textit{$\kappa$-multilinear map} $e : G_1\times \dots \times G_\kappa \rightarrow G_T$ has the following properties:

(1) For elements $g_1 \in G_1,\dots, g_\kappa \in G_\kappa$, index $i \in \left[ \kappa \right]$ and an integer $\alpha \in \mathbb{Z}_p $, we have:
 \[e(g_1,\dots,\alpha \cdot g_i,\dots,g_\kappa)= \alpha\cdot e(g_1,\dots,g_\kappa)\]

(2) The map $e$ is non-degenerate, which means that if $g_i \in G_i\; (i=1,\dots,\kappa) $  is a generator of $G_i $, then $e(g_1,\dots,g_\kappa)$ is a generator of the target group $G_T$.

In the above definition, if $G_i \; (i=1,\dots,\kappa)$ are all identical groups, it is called a symmetric multilinear map.
\end{definition}

A cryptographic multilinear map has been a long sought after and powerful tool. Not until recently was it approximately constructed in the form of  Graded Encoding System \cite{garg2013candidate,coron2013practical}. In particular, we will design our algorithms based on the construction over integers, which is a more practical one devised by Coron, Lepoint, and Tibouchi\cite{coron2013practical}. Now we briefly recall the definition of Graded Encoding System (GES) and its associated efficient procedures for manipulating encodings on which the description of our algorithms is based directly.
Note that we only consider the symmetric case here.
%in Definition \ref{Def2} and Definition \ref{Def3}.

\begin{definition}(\textbf{$\kappa-$Graded Encoding System}\cite{garg2013candidate})
\label{Def2} A $\kappa$-Graded Encoding System consists of a ring $R$ and a system of sets $\mathcal{S} = \lbrace{S^{(\alpha)}_v \in {\lbrace 0,1\rbrace}^* : v \in \mathbb{N}, \alpha \in R \rbrace}$, with the following properties:

1. For every $v \in \mathbb{N}$, the sets $\lbrace S^{(\alpha)}_v : \alpha \in R \rbrace$ are disjoint.

2. There is an associate binary operation `+' and a self-inverse unary operation `$-$' (on ${\lbrace 0,1 \rbrace}^*$) such that for every $\alpha_1, \alpha_2 \in R$, every index $i \leq \kappa$, and every $u_1 \in S^{(\alpha_1)}_i$ and $u_2 \in S^{(\alpha_2)}_i$, it holds that
\[ u_1 + u_2 \in S^{(\alpha_1 + \alpha_2)}_i \; , \; -u_1 \in S^{(-\alpha1)}_i\]
where $\alpha_1 + \alpha_2$ and $-\alpha_1$ are addition and negation in $R$.

3. There is an associate binary operation `$\times$' (on ${\lbrace 0,1 \rbrace}^*$) such that for every $\alpha_1, \alpha_2 \in R$, every $i_1, i_2$ with $i_1 + i_2 \leq \kappa$, and every $u_1 \in S^{(\alpha_1)}_{i_1}$ and $u_2 \in S^{(\alpha_2)}_{i_2}$, it holds that $u_1 \times u_2 \in S^{(\alpha_1 \cdot \alpha_2)}_{i_1 + i_2}$. Here $\alpha_1 \cdot \alpha_2$ is multiplication in $R$, and $i_1 + i_2$ is integer addition.
\end{definition}

%Then we recall some of
\begin{definition}(\textbf{Efficient Procedures for $\kappa$-Graded Encoding System}\cite{garg2013candidate,coron2013practical})
\label{Def3} For graded encoding system, we have the following efficient procedures:

\textbf{Instance Generation}:$(params,P_{zt})\leftarrow\mathnormal{InstGen}(1^\lambda, 1^\kappa)$, where $params$ is a description of a $\kappa$-Graded Encoding System with security parameter $\lambda$, and $P_{zt}$ is a zero-test parameter for level $\kappa$.

\textbf{Ring Sampler}: $c \leftarrow \mathnormal{samp(params)}$. This procedure outputs a ``level-zero encoding" $c \in S^{(\alpha)}_0$ for a nearly uniform element $\alpha \in_R R$.

\textbf{Encoding}: $c_k \leftarrow \mathnormal{enc(params,k,c)}$. This procedure outputs the ``level-i encoding" $c_k \in S^{(\alpha)}_i$ for a ``level-zero encoding" $c \in S^{(\alpha)}_0$.

\textbf{Re-Randomization}: $c' \leftarrow \mathnormal{reRand(params,i,c)}$. Procedure $reRand$ can re-randomize encodings relative to the same level i.

\textbf{Addition and negation}: $u' \leftarrow \mathnormal{add(params,i,u_1,u_2)}$ and $u' \leftarrow \mathnormal{neg(params,i,u_1)}$. The two procedures are corresponding to operation `+' and `$-$' in the above definition.

\textbf{Multiplication}: $u' \leftarrow \mathnormal{mul(params, i_1, u_1, i_2, u_2)}$. This procedure is for operation  `$\times$' in the above definition.

\textbf{Zero Testing}: $\mathnormal{isZero(params,P_{zt},u_{\kappa}) \overset{?}{=} 0/1}$. This procedure will check whether $u_{\kappa} \in S^{(0)}_{\kappa}$.

\textbf{Extraction}: $sk \leftarrow \mathnormal{ext(params, P_{zt},u_{\kappa})}$. This procedure extracts a ``canonical'' and ``random" representation of ring elements from their level-$\kappa$ encoding.

\end{definition}

\section{SOFA: Framework and A Basic Scheme}
\label{sec4}
As a concrete example for secure network middlebox outsourcing, we present SOFA  in detail in this section. First, we will introduce our framework which is comprised of two phases, the obfuscating phase and the online execution phase. Then we give a naive scheme by directly applying the theoretical conjunction obfuscation technique \cite{zvika2013obfus}. The discussion of its lack of feasibility will lead to our basic scheme, which is much more efficient in obfuscating the original firewall. And in the next section, we will further present another two approaches  to boost our scheme's efficiency in execution phase.

%Note that we do not consider the overall architecture for middleboxes outsourcing in SDN here. The overall architecture, of which the main goal is to design the configuration negotiation mechanism, to optimize the traffic redirection manner and so on, has been presented, such as APLOMB\cite{sherry2012making} and Jingling\cite{gibb2012outsourcing}. Besides, our scheme SOFA is compatible with them, which means that we can ``insert" our scheme easily into this overall NFO architecture as a module by taking advantage of the central control methodology in SDN. So our scheme presented here focuses on the privacy preserving outsourcing requirement while the other related tasks such as traffics redirection  and resource management are beyond the scope of our work.

\subsection{Fundamental Framework}

A \textit{cryptographic obfuscator }of a program can make a program perform the same functionality while the detailed underlying instructions are hidden \cite{barak2012possibility}.
Combining the idea of \textit{program obfuscation}  and the Definition \ref{PPFP}, we observe that the privacy preserving firewall problem can be solved by constructing a cryptographic obfuscator $\mathcal{O}$ for firewall \textit{f} such that \textit{f\'}=$\mathcal{O}(f)$. The secrets are the underlying firewall rules, and they can be hidden in \textit{f\'}. Generally, this relationship between privacy preserving outsourcing requirement and cryptographic obfuscator construction can be extended to other middleboxes.

Our framework  SOFA consists of two phases. The first phase is the construction of an obfuscator for original firewall \textit{f}. And this work is done by the control plane of SDN in the local enterprise. The second phase is the execution of the obfuscated firewall \textit{f\'}=$\mathcal{O}(f)$  in the cloud. And this task should be done by the control plane of cloud firewall provider.
%, such as the ISP for the local enterprise networks.

\textbf{Obfuscation Phase}: In the first phase, the controller of the enterprise SDN can set up the basic cryptographic parameters of the multilinear map. And then for each rule \textit{r} in the firewall, such as these in Table \ref{table1}, the control plane generates a corresponding atomic rule ``ciphertext" \textit{r\'}. Subsequently, the control plane can aggregate these atomic ciphertexts to be an integral secure firewall \textit{f\'}. And at last it sends  \textit{f\'} and some necessary public parameters to the cloud provider. This obfuscation phase for firewall can be done offline and once for all. And if there are some updates for the firewall \textit{f}, we can just reconstruct the related atomic rules, reassemble them and then replace the old ones in \textit{f\'} with them. But this process still should be done as efficiently as possible since one main purpose of outsourcing is to reduce costs for enterprises.

\textbf{Execution Phase}: After receiving \textit{f\'}, the cloud firewall provider performs the packet filtering task on the inbound and outbound traffics of the protected enterprise. For each packet, the high-performance devices in the cloud, under the supervision of its controller, securely checks its packet header according to \textit{f\'}.  If the packet is found satisfying a rule \textit{r\'},  the cloud firewall devices will take the corresponding actions associated with this rule, such as  permitting it,  or denying it.
% or logging this event.

The firewall obfuscation is implemented with cryptographic multilinear map. So in the following we  model the firewall rule and traffic packet header in the bitwise level.  For each rule, we present it in the format of $R = (\vec{v},W,A)$, where $\vec{v} \in {\lbrace 0,1\rbrace}^n$ specifies the expected bit, 0 or 1, on non-wildcard bits, the set $W \subseteq \left[ n \right] $ specifies the  ``wildcard'' bits, and $A$ indicates the associated actions. Note that actually set $W$ is often the subnet mask in firewall rules. For ease of presentation, we denote $\vec{v} \left[ i \right]=0$ if $i \in W$.  For instance, the rule \#1 in Table \ref{table1}, ``192.168.45.*'', can be presented as $\vec{v}=``1100000010101000..."$,
%$\vec{v}=(1,1,0,0,\dots,0)$
 $W={\lbrace 25,26,\dots,32\rbrace}$ and $A=\lbrace deny \rbrace$. In addition, $A$ is the output decision on a packet, it leaks little information about the detailed enterprise network, so we do not consider protecting it in this paper. Please see section \ref{sec5} for a detailed analysis of our schemes' security. Similarly, for a tarffic packet's TCP/IP header, we also present its related bits  as a n-bit vector $\vec{p} \in {\lbrace 0,1\rbrace}^n$.

\subsection{A Naive Scheme}
A naive and intuitive idea is to directly apply the theoretical construction of conjunction obfuscator\cite{zvika2013obfus} in obfuscating the firewall rules.
Here we give a concise description of this scheme. For each rule  \textit{r} in firewall \textit{f}, we present it as $\textit{r} \triangleq (\vec{v},W,A)$. Then for the $i$-th bit of it, we generate two pairs of level-1 encodings in the GES (We term them ``encoding pairs unit"):
\begin{align}
(u_{i,0},v_{i,0})&\in({S^{\rho_{i,0}}_1, S^{\rho_{i,0} \times \alpha_{i,0}}_1}), \label{Eq1}\\
(u_{i,1},v_{i,1})&\in({S^{\rho_{i,1}}_1, S^{\rho_{i,1} \times \alpha_{i,1}}_1})\label{Eq2}
\end{align}

In Equation (\ref{Eq1}) and (\ref{Eq2}), $\rho_{i,0},\alpha_{i,0},\rho_{i,1},\alpha_{i,1}$ are all randomly and uniformly chosen if  $i \notin W$, but for $i \in W$, they should also satisfy $\alpha_{i,0} = \alpha_{i,1}$. In addition to two pairs for each bit, an encoding pair is also generated for a whole rule:
\begin{equation}\label{Eq3}
(u_{n+1},v_{n+1}) \in(S^{\rho_{n+1}}_1,S^{\rho_{n+1}\times\prod_{i \in \left[ n \right] }\alpha_{i,\vec{v}\left[ i \right]}}_1)
\end{equation}

In the execution phase, for each obfuscated rule, we pick and multiply  encodings from each pair according to $\vec{p}$ and get the following two (n+1)-level encodings:
\begin{equation}\label{Eq4}
%LHS&=S^{\big(\rho_{n+1} \times (\prod_{i \in \left[ n \right] } \rho_{i,\vec{p}[i]} \times \alpha_{i,\vec{p}[i]} )\big)}_{n+1}\\
%RHS&=S^{\big( (\rho_{n+1}\times \prod_{i \in \left[ n \right] }\alpha_{i,\vec{p}\left[ i \right]}) \times ( \prod_{i \in \left[ n \right] } \rho_{i,\vec{p}[i]})\big)}_{n+1}
LHS=u_{n+1}\times \prod_{i \in \left[ n \right] } v_{i,\vec{p}\left[ i \right]},
RHS=v_{n+1}\times \prod_{i \in \left[ n \right] } u_{i,\vec{p}\left[ i \right]}
\end{equation}

At last procedure $\mathnormal{isZero}$ is performed to determine whether $LHS$ (Left-hand side) and $RHS$ (Right-hand side) are equal to judge whether packet $\vec{p}$ satisfies this rule and which actions to take.

Actually, the scheme we briefly presented above is the simplified symmetric counterpart for the original asymmetric construction in \cite{zvika2013obfus}, and the asymmetry is vital in its theoretical security proof to defend the DDH attack which may leak the wildcard locations. It is much harder and inefficient to construct such an asymmetric multilinear map. More importantly, as we all know, the wildcard part in a IP address is for indicating the subnet range, and these bits are always continuous in the later part. Nothing additional useful information about specific subnets or hosts, except for the the size of an unknown subnet, can be discovered without $\vec{v}$. So leaking this subnet mask rather than the specific prefixes ( $\vec{v}$ ) is acceptable in most firewall scenarios. Thus in this paper, we do not attempt to perfectly defend such DDH attack and herein we use symmetric multilinear map in all our schemes and leave such protection as future work. But we point out that following schemes are also fit for asymmetric one.

 \textbf{Discussion}.  Now we argue that this naive scheme will have serious issues in practical application.

 (1) Its time and space consumption in the obfuscating phase will be too large. If we have $l$ rules in a firewall, then the obfuscation will take $\Theta(ln+l)$ times encoding and re-randomization, and need store $\Theta(ln+l)$ encoding ciphertexts. In addition, when $n$ is large, the singe procedures will take more time and the storage will also be larger. So such privacy preserving scheme will be impractical due to the computation and storage cost in obfuscating phase is hardly acceptable.

 (2) When $n$ is large, it also will be much harder to construct the GES due to the noise limit. Thus to make the $2(n+1)$ multiplication operations and $\mathnormal{isZero}$ viable, the underlying GES parameters have to be set large enough. Consequently the real-time execution in cloud will be too inefficient to have any practical usability if we adopt the naive scheme directly.

Therefore we have to make up for these limitations to make it more practical.

\renewcommand{\algorithmicrequire}{\textbf{Input:}}
\renewcommand{\algorithmicensure}{\textbf{Output:}}

\begin{algorithm}[!t]
%\SetLine
\caption{Basic Scheme For Obfuscating Firewall}
\label{Basic_phase1}%{intuitive}
\begin{algorithmic}%[1]
\REQUIRE A original plain firewall \textit{f} with rule sets $\lbrace r\rbrace $.
\ENSURE A obfuscated firewall \textit{f\'} with the same functionality.

\STATE
$(params,P_{zt})\leftarrow\mathnormal{InstGen}(1^\lambda, 1^{n+1})$.

\FOR {$i \in \{1,\ldots,M+N\}$}
	\STATE
	$C_{\rho_{i,0}} \leftarrow \mathnormal{samp(\cdot)} ,\quad C_{\rho_{i,1}}\leftarrow \mathnormal{samp(\cdot)}$.
	\IF{$i \leq M$}
	
	\STATE $C_{\alpha_{i,0}} = C_{\alpha_{i,1}} \leftarrow \mathnormal{samp(\cdot)}$.
	\ELSE
	\STATE $C_{\alpha_{i,0}}\leftarrow \mathnormal{samp(\cdot)},\quad C_{\alpha_{i,1}} \leftarrow \mathnormal{samp(\cdot)}$.
	\ENDIF
	
	$C_\beta\leftarrow enc(\cdot,1,C_{\rho_{i,0}})$,
	$u_{i,0}\leftarrow reRand(\cdot,1,C_\beta)$.
	
	$C_\gamma \leftarrow enc(\cdot,1,C_{\rho_{i,0}} \times C_{\alpha_{i,0}})$,
	$v_{i,0} \leftarrow reRand(\cdot,1,C_\gamma)$.
	
	$C_\delta \leftarrow enc(\cdot,1,C_{\rho_{i,1}})$,
	$u_{i,1} \leftarrow reRand(\cdot,1,C_\delta)$.
	
	$C_\epsilon \leftarrow enc(\cdot,1,C_{\rho_{i,1}} \times C_{\alpha_{i,1}})$,
	$v_{i,1} \leftarrow reRand(\cdot,1,C_\epsilon)$.
	
	$C[i] \triangleq \lbrace(u_{i,0},v_{i,0}),(u_{i,1},v_{i,1})\rbrace$.
\ENDFOR
\STATE $C$=$\sigma(C)$.
\STATE $E=\lbrace \sigma(1),\ldots,\sigma(M)\rbrace$,
\STATE $UE=\lbrace \sigma(M+1),\ldots,\sigma(M+N)\rbrace$.
\FOR { $\textit{r} \triangleq (\vec{v},W,A) \in$ \textit{f} }

\FOR {$i \in \{1,\ldots,n\}$}
	\IF{$i \in W$}
	\STATE $P_r[i]\overset{R}{\nleftarrow} E$.
	\ELSE
	\STATE $P_r[i]\overset{R}{\nleftarrow} UE$.
	\ENDIF
		
\ENDFOR

%\FOR {$i \in \{1,\ldots,n\}$}
 \STATE $\rho_{i,n+1}\leftarrow \mathnormal{samp(\cdot)}$,
 \STATE  $C_\theta \leftarrow enc(\cdot,1,C_{\rho_{i,n+1}})$,
	$u_{n+1} \leftarrow reRand(\cdot,1,C_\theta)$.

 $C_\mu \leftarrow enc(\cdot,1,C_{\rho_{i,n+1}} \times\prod_{i \in \left[ n \right] }C[P_r[i]]. C_{\alpha_{i,\vec{v}\left[ i \right]}})$,%\footnotemark{},

 $v_{n+1} \leftarrow reRand(\cdot,1,C_\mu)$.
%\ENDFOR

\STATE  \textit{r\'} $\triangleq$
	$\big(P_r,u_{n+1},v_{n+1}, A\big)$.
\ENDFOR
\RETURN  \textit{f\'}$\triangleq \big( {\lbrace C[i] \rbrace}_{i \in [M+N]}, \lbrace r' \rbrace, P_{zt}, params'\big)$.
\end{algorithmic}

%\end{figure}
\end{algorithm}

\subsection{Our Basic Scheme}
\label{OC}

Here, we first give a basic scheme by compressing the encodings to address the issue (1) discussed above. The computation and space complexity will be reduced from $\Theta(ln+l)$ to $\Theta(n+l)$ encodings. This scheme is not as secure as the naive one, but the amount of  leaked information is limited and acceptable ( please see section \ref{sec5} for details).

In the naive scheme, the encoding pairs $(u_{i,0},v_{i,0})$ and $(u_{i,1},v_{i,1})$ for every rule are  dependent on $W$ while the pair $(u_{n+1},v_{n+1})$ is dependent on $\vec{v}$. Now we give our basic scheme, in which the encoding pairs for each bit every rule are extracted and shared.
%\footnotetext{Here. How to add footnote for algorithmic environment?}

The obfuscating steps are summarized as Algorithm \ref{Basic_phase1}\footnote{For conciseness, we denote the parameter $params$ as `` $\cdot$ " in every primitive procedure in all our following algorithms. In addition, The $C$ we outsource do not contains $C_{\alpha_{i,\vec{v}\left[ i \right]}}$, but in the obfuscating phase, we use $C[P_r[i]]. C_{\alpha_{i,\vec{v}\left[ i \right]}}$ to denote the access for presentation convenience.} utilizing primitive procedures in GES. We first get $M+N$ encoding pairs units, of which the first $M$ units have equal random ratios in two encoding pairs while in the last $N$ units such ratios are different. Then we permute ($\sigma$) these encoding pairs units to shuffle them, and then collect two indexes sets:
\begin{align}
E&=\lbrace \sigma(1), \sigma(2),\ldots,\sigma(M)\rbrace,\\
UE&=\lbrace \sigma(M+1), \sigma(M+2),\ldots,\sigma(M+N)\rbrace.
\end{align}
Set $E$ collects the indexes of encoding pairs units with $\alpha_{i,0}= \alpha_{i,1}$ while set $UE$ collects those with $\alpha_{i,0}\neq \alpha_{i,1}$.

To obfuscate a rule $\textit{r} \triangleq (\vec{v},W,A)$, we attach an encoding pairs unit indexes array $p_r$ to \textit{r\'} such that if $i \in W$, then $p_r[i]$ is randomly sampled without replacement ($\overset{R}{\nleftarrow}$) from set $E$, otherwise it is randomly sampled without replacement from set $UE$. Thus $p_r$ specifies the $n$ encoding pairs units for rule $r$. Then $(u_{n+1},v_{n+1})$ for this rule is calculated according to Equation (\ref{Eq3}) by choosing encoding pairs in $C[P_r[i]]$ with $\vec{v}$. At last, public parameters $params'$ of GES and zero-test parameter $P_{zt}$ are also included in the obfuscated firewall \textit{f\'}. We emphasize that we should sample randomly without replacement rather than directly sampling with replacement. This trick is used to eliminate the possibility of false positive which may happen when two non-wildcard bits with different values chose the same encoding pair unit.

\begin{algorithm}[!t]
\caption{Execution with Outsourced Firewall \textit{f\'} }
\label{Basic_phase2}
\begin{algorithmic}%[1]
\REQUIRE A packet $\vec{p}={\lbrace 0,1\rbrace}^n$, and obfuscated firewall \textit{f\'}.
\ENSURE Decisions on $\vec{p}$.

\FOR { \textit{r\'} $\triangleq$
	$\big(P_r,u_{n+1},v_{n+1}, A\big) \in $ \textit{f\'} }
	\STATE
	
	 $LHS=u_{n+1}\times \prod_{i \in \left[ n \right] } C[P_r[i]].v_{i,\vec{p}\left[ i \right]} $,
	
	 $RHS=v_{n+1}\times \prod_{i \in \left[ n \right] } C[P_r[i]].u_{i,\vec{p}\left[ i \right]}$.

	 \IF {$\mathnormal{isZero}(params',P_{zt},LHS - RHS)$}
	 	\RETURN $A$ and exit.
	 \ELSE
	 	\STATE \textit{Continue} checking $\vec{p}$ with the next rule.
	 \ENDIF
	
\ENDFOR
\end{algorithmic}
\end{algorithm}

In the second phase, the execution of  outsourced firewall \textit{f\'} is straightforward. The basic  steps are summarized in Algorithm \ref{Basic_phase2}. For an input packet $\vec{p}$, we check whether it satisfies an underlying rule in \textit{f\'}  sequentially. Note that for each obfuscated rule \textit{r\'}, the multiplications are also performed on encodings chose from encoding pairs units $C$ specified by $\vec{p}$ and $P_r$.

Our scheme's correctness can be verified easily. In Algorithm \ref{Basic_phase1}, only two encodings are generated to a specific rule. The overheads of $P_r$ generation is negligible compared to encoding and re-randomization, thus it is easy to figure out  the complexity of Algorithm \ref{Basic_phase1} is $\Theta(n+l)$. Thus the basic scheme is efficient in the obfuscation phase. The efficiency in execution phase is not bad, but needs further improvement.

%We note that for a rule $\textit{r}$ with $m=n-\vert W \vert$ non-wildcard bits, consisting of $k$ ``1" bits and $m-k$ ``0" bits, the false positive probability for our scheme to  misclassify a packet $\vec{p}$ will be $\frac{k!(m-k)!}{N^m}$ due to the sharing $C$, which will be very low even when $N=n$.
%So as we can see, the encoding pairs are now shared

\section{Online Efficiency Improvement}
\label{Improve}
In this section, we try to improve the efficiency and usability of our scheme in the online execution phase. As we have discussed previously, the online efficiency is vital for middlebox outsourcing scheme. The two enhanced schemes we propose here can reduce its complexity from $\Theta(n)$ to $\Theta(k)$ or $\Theta(n/k)$ on faster encodings' multiplication operations. The cost is $\Theta(1)$ increasement in encoding and re-randomization in the  first phase.

The main guideline for our improvement is to reduce overheads of single multiplication and the times of multiplication. To this end, in both of the following approaches, we partition an original bit-wise rule into several parts. In the blocking approach, we then construct an encoding pair for an element in each part rather than for each rule bit, and thus we can check a packet in block level rather bit level. In the divide-and-conquer approach, we then test if a packet satisfies the rule in each part separately and sequentially with less encoding levels.

\subsection{Blocking Scheme}

We find that if we view a firewall rule in a higher level rather than the primitive bit level, there will be fewer occurrences of multiplication needed in online phase. If the block is small enough, the single multiplication will also take less time. To this end, we take another approach of firewall obfuscation by constructing in block level rather than bitwise level.
%to decrease the max levels of encodings in our scheme.
The basic steps in the first phase are described as algorithm \ref{practical_phase1}.

\begin{algorithm}[!t]
\caption{ Firewall Obfuscation (Blocking Approach)}
\label{practical_phase1}
\begin{algorithmic}%[1]
\REQUIRE A original plain firewall \textit{f} with rule sets $\lbrace r\rbrace $.
\ENSURE A obfuscated firewall \textit{f\'} with the same functionality.

\STATE Generate parameters:

$(params,P_{zt})\leftarrow\mathnormal{InstGen}(1^\lambda, 1^{k+1})$.

\FOR { $\textit{r}\triangleq ({\lbrace F_i\rbrace}_{i \in [k]}, A) \in$ \textit{f}}

	\FOR {$i \in \{1,\ldots,k\}$}
		\STATE
		$C_{\eta_i} \leftarrow \mathnormal{samp(\cdot)}$.

		\FOR { $j \in I_i$}
			\STATE
			$C_{\beta_j} \leftarrow \mathnormal{samp(\cdot)}$
	
			\IF{$j \in S_i$}	
				\STATE  $C_{\alpha_j} = C_{\eta_i}$.
			\ELSE
				\STATE
     			$C_{\alpha_j} \leftarrow \mathnormal{samp(\cdot)}$.
			\ENDIF		
		
			\STATE
			$C_\gamma \leftarrow enc(\cdot,1,C_{\beta_j})$,
			$C_\delta \leftarrow enc(\cdot,1,C_{\beta_j}\times C_{\alpha_j})$,
			
			$u_{i,j} \leftarrow reRand(\cdot,1,C_\gamma)$, 		
	 		$ v_{i,j} \leftarrow reRand(\cdot,1,C_\delta)$.
		
		\ENDFOR
	\ENDFOR

	\STATE
 	$C_{\rho_{k+1}}\leftarrow \mathnormal{samp(\cdot)}$.

 	\STATE
	$C_\epsilon \leftarrow enc(\cdot,1,C_{\rho_{k+1}})$,
	
	$u_{k+1} \leftarrow reRand(\cdot,1,C_\epsilon)$,
	
	 $C_\theta \leftarrow enc(\cdot,1,C_{\rho_{k+1}}\times\prod_{i \in \left[ k \right] }C_{\eta_i})$,
	
	 $v_{k+1} \leftarrow reRand(\cdot,1,C_\theta)$.
	
	\STATE  \textit{r\'} $\triangleq
	\big({{\lbrace ( u_{i,j},v_{i,j} )\rbrace}}_{i \in \left[ k \right], j \in I_i}, u_{k+1},v_{k+1}, A  \big)$.

\ENDFOR

\RETURN  \textit{f\'} $\triangleq \big(\lbrace r' \rbrace, P_{zt}, params' \big)$.
\end{algorithmic}
\end{algorithm}

Note that for brevity and highlighting the key idea, this algorithm is described based on the naive scheme. The encoding pairs compression idea in basic scheme can be easily integrated in algorithm \ref{practical_phase1}. We omit this integration here.

Specifically, in the first phase, for a firewall rule, we split it into $k$ fields, each with a small integer range. Take the second rule in table \ref{table1} as an example, we can split its IP source address header into four fields, each with the same domain $I_{1,2,3,4}=[0,255]$ (Just as its dotted decimal notation). In each field, the rule specifies a target filter set, such as  $F_{1}={192},F_{2}={168},F_3={45}$ and $F_4=[0,255]$ for the above example. Another advantage of this approach is that it enables us to process some rare extended ACLs with port range. For example, for the destination port range in the rule \#3 in table \ref{table1}, we can directly set the corresponding target filter set as $[22,88]$. So now a rule can be presented as $\textit{r}\triangleq ({\lbrace F_i\rbrace}_{i \in [k]}, A)$.

After this partition, for each field $i$ within the domain $I_i$, we ``encrypt" each element in it with a level-1 encoding pair. The key point is that, for all elements that are in the filter set $F_i$, the underlying ratios of the two encodings in their encoding pairs are all equal to a constant number $\eta_i$,  which is  specified for each field in this obfuscation. Finally, just similar to  basic scheme, these ratios in each field are aggregated in a product of them, which is also hidden as a ratio in another level-1 encoding pair $(u_{k+1},v_{k+1})$:
\begin{equation}
(u_{k+1},v_{k+1}) \in (S^{\rho_{k+1}}_1,S^{\rho_{k+1} \times \prod_{i \in \left[ k \right] }\eta_i}_1)
\end{equation}

%We also get a cipher pair corresponding to element $\prod_{i \in \left[ k \right] }\eta_i$ in another source group. What's critical is that, for the elements that are also in target range $S_i$, the ratios of the encryptions in each cipher pair are all the same (in algorithm 2A, this ratio is set as predefined constant $\eta_i$). And we also get a cipher pair corresponding to element $\prod_{i \in \left[ k \right] }\eta_i$ in another source group.

Then, in the second phase, as a consequence of the simplification in the first phase, the burden of the cloud firewall provider will be greatly relieved. For a packet $\vec{p}$, cloud firewall provider also splits its packet header into $k$ blocks, which means that we can present it in the format of an integer tuple $(p_1,\dots,p_k)$. Then an encodings pair in each field is chosen according to the integer tuple value followed by the computation of two level-($k+1$) encodings,$LHS$ and $RHS$, by multiplying them.  At last the equality of $LHS$ and $RHS$ is tested and thus a decision upon this packet can be made. The modification of algorithm 2 to coordinate with algorithm 3 is easy, so we leave out the details here.

\subsection{Divide-and-Conquer Scheme}

To reduce overheads of a single multiplication, we can decrease the maximum levels in GES. So another strategy we can take is divide-and-conquer. We divide the whole bit-string rule $(\vec{v},W, A)$ into $k$ parts: $(\vec{v},W,A)=(\lbrace\vec{v}_i,W_i\rbrace_{1\leq i\leq k},A)$. And in each part, we construct a obfuscated $r'_i$ in the same way as algorithm 1. Thus the corresponding obfuscated firewall will be  $f'\triangleq \big( \big \lbrace  {\lbrace r'_i \rbrace}_{1\leq i\leq k} \big\rbrace, {\lbrace P_{zt}^i\rbrace}_{1\leq i\leq k}, params'\big)$. Similarly, in the second phase, we also partition the packet in the same way, and get $\vec{p}=(\vec{p}_{i})_{1\leq i\leq k}$. Then we can check whether $\vec{p}$ matches the underlying $r$ by checking $\vec{p}_{i}$ with  $r'_i$ in each part sequentially with algorithm 2.

This approach is simple, but effective, as will be demonstrated in our experiments. It is not hard to modify the basic scheme in this approach, so we omit the details here.

%Besides the above improvement by blocking the rule and packet bit-string, there are other simple, yet effective, tricks can be used to further boost the efficiency of SOFA. For instance, not all bits in rules and packet headers are equally important, thus some of them can be left out without being included in our construction and execution of firewall \textit{f\'}.

%\subsection{Bilinear Map Based Simplification}
%In this subsection, we propose a more efficient algorithm, which is based on bilinear map rather than multilinear map. We also inherit the idea of partitioning in algorithm 2. This construction is presented as algorithm 3 in Fig.\ref{FirewallObfs_bilinear}. Because the construction and computation in bilinear map are comparatively much more efficient than in multilinear maps, this algorithm is more viable notwithstanding this algorithm is not as secure as algorithm 1 nor algorithm 2.
%
%The steps in algorithm 3 are similar to these in algorithm 2, except that the encryptions are performed in the source group of bilinear map. In the second phase, due to the limitations of bilinear map, it is unable to aggregate the values on each block to get the product of them while the equality still can be checked directly. So we need to check the equality block-wise. That is to say, for packet $x$ and the $i$-th block, we need to determine whether $e(C1_{i,x_i},T2_i)$ is equal to $e(C2_{i,x_i},T1_i)$. If the answer is yes, we continue to check the next block. Otherwise we know this packet dose not hit this rule, thus we break to check in obfuscator for the next rule.

\section{privacy analysis}
\label{sec5}
In this section, we will take a close look at the security of the privacy preserving schemes we proposed.
\subsubsection{Analysis of the basic scheme}
In \cite{zvika2013obfus}, the authors have proved the naive scheme is secure when the inputting function is drawn from a distribution with a high superlogarithmic  entropy given the GXDH assumption  and the GCAN assumption \cite{zvika2013obfus} are true or in the \textit{generic model} where the adversary is only allowed to manipulate encodings via oracle access.

When parameters $M$ and $N$ are large enough, our basic scheme achieves the same security as the naive scheme and the security proof is roughly the same as proofs provided in \cite{zvika2013obfus}. Due to both completeness considerations and the page limit, below we present the proof sketch for the high entropy input distribution case only and neglect the proof in the generic model.
\begin{theorem}\label{thm:alg12sec}
When $M$ and $N$ are large enough so that all $P_r[i]$ are different, Algorithms 1 and 2 are secure when $(\vec{v},W)$ is drawn from a distribution where the entropy of $\vec{v}$ given $W$ is superlogrithmic given the GXDH assumption and the GCAN assumption are true.
\end{theorem}
\begin{proof}
Same as \cite{zvika2013obfus}, by claiming Algorithm 1 and 2 are secure we mean that, for each rule $r$, anything that the cloud firewall provider or the adversary learns from the execution of the two algorithms can be learned from a blackbox firewall by constructing a simulator  $\mathcal{S}$ that simulates the adversary's view from algorithms 1 and 2 given only the access of a blackbox firewall.

In particular, the adversary's view is the obfuscated firewall output by Algorithm 1:$\big( {\lbrace C[i] \rbrace}_{i \in [M+N]}, \lbrace r' \rbrace, P_{zt}, params'\big)$.
$\mathcal{S}$ simulates $ {\lbrace C[i] \rbrace}_{i \in [M+N]}$ with $4(M+N)$ uniformly random level-1 encodings: $(\{u'_{i,0},v'_{i,0}, u'_{i,1},v'_{i,1}\}_{i\in[M+N]})$. $S$ first runs \textit{InstGen} and get a random $P'_{zt}$ and random public parameters $params''$. For each $r'$, $S$ can simulate $P_r$ with $n$ different random samplings in $[M+N]$ ,  simulate $(u_{n+1},v_{n+1})$ with two uniformly random level-1 encodings: $(u'_{n+1},v'_{n+1})$, simulates $A$ with $A$ itself, and simulates $P_{zt}$ and $params'$ with $P'_{zt}$ and $params''$. Due to the GCAN assumption and the entropy of $\vec{v}$ given $W$ is superlogrithmic, we know  $u_{n+1},v_{n+1}$ is indistinguishable from uniform thus is independent with any pair in $\{C[i]\}$. In addition, due to the GXDH assumption, we know each $C[i]$ is indistinguishable from $(u'_{i,0},v'_{i,0}, u'_{i,1},v'_{i,1})$, and $(u_{n+1},v_{n+1})$ is indistinguishable from $(u'_{n+1},v'_{n+1})$. Thus we know $\mathcal{S}$ simulates the algorithm 1's output for ($(\vec{v},W)$) input.
For algorithm 2's output, the $S$ can simply use the access of the blackbox firewall to get it.
\end{proof}

When $M$ and $N$ are not large enough, Algorithm 1 reveal partial information about the rules. The leakage could happen
%in two cases: 1) For each rule, there exists different $i,j\in \{1,\ldots,n\}$ such that $P_r[i]=P_r[j]$; 2)
when for different rules $r_1$ and $r_2$, there exist $i,j\in \{1,\ldots,n\}$ such that $P_{r_1}[i]=P_{r_2}[j]$.
%In this first case, the adversary would notice two same encodings corresponding to two bit locations inside one rule. This allows the adversary to know that either the rule's two bit locations are both wildcards or neither of them is.
In this case, the adversary would notice similar ``wildcard or not'' information between two bit locations in two rules.
%For the ease of presentation, we refer to these two kinds of information leak as the Type I leak and Type II leak.
We can figure out that the probability for such leakage is
$1-\big( \frac{\binom{M}{|W_1|} \cdot \binom{M-|W_1|}{|W_2|}} {\binom{M}{|W_1}\cdot \binom{M}{|W_2|}} \big) \cdot
 \big( \frac{\binom{N}{m_1} \cdot \binom{N-m_1}{m_2}} {\binom{N}{m_1}\cdot \binom{N}{m_2}} \big)
 = 1- \big( \frac{(M-|W_1|)!\cdot (M-|W_2|)!}{M!\cdot (M-|W_1|-|W_2|)!}\big) \cdot
\big( \frac{(N-m_1)!\cdot (N-m_2)!}{N!\cdot (N-m_1-m_2)!}\big)
$, in which the $W_1$, $W_2$ is the wildcard bits set in $r_1$ and $W_2$ each, and $m1=n-|W_1|, m2=n-|W_2|$. So if such leak is unacceptable, we may increase the value of $M$ and $N$ in practice.

%$1-\frac{M!}{M^{|W|}\cdot(M-|W|)!} \cdot \frac{N!}{N^{n-|W|}\cdot (N-n+|W|)!}
%$. And when $r_1$ takes $s_1$ encoding pairs from $E$, $t_1$ from $UE$, and $r_2$ takes $s_2$ encoding pairs from $E$, $t_2$ from $UE$, the probability for type II leak is $1-\frac{M!}{M^{s_1+s_2}\cdot (M-s_1-s_2)!}\cdot\frac{N!}{N^{t_1+t_2}\cdot (M-t_1-t_2)!}$

\subsubsection{Analysis of the Blocking Scheme}
The main difference between the basic and blocking scheme is that Algorithm 3 ``encrypt"s each integer in each field rather than ``0" or ``1" on each bit.
For similar reasons as Algorithm 1, we also give a proof sketch for Algorithm 3's security in the case of high entropy inputs and leave out the detailed proofs. Note that  Algorithm 3 itself achieves the same security as the naive scheme's security. So when integrated with basic scheme, the blocking scheme is as secure as Algorithm 1.
\begin{theorem}Algorithm 3 is secure when $(\vec{v},W)$ is drawn from a distribution where the entropy of $\vec{v}$ given $W$ is superlogrithmic given the GXDH assumption and the GCAN assumption are true.
\end{theorem}
\begin{proof}
The proof goes similarly as the Proof of Theorem \ref{thm:alg12sec}. The only difference here is $\mathcal{S}$ simulates algorithm 3's output $\big({{\lbrace ( u_{i,j},v_{i,j} )\rbrace}}_{i \in \left[ k \right], j \in I_i}, u_{k+1},v_{k+1}\big)$ with $2\Sigma^k_1{|I_i|}+2$ uniformly random level-0 encodings: $(\{u'_{i,j},v'_{i,j}\}_{i\in[n],j\in I_i}, u'_{k+1},v'_{k+1})$. Due to the GXDH assumption, we know for each $i\in[n]$, $\{<u'_{i,j},v'_{i,j}>\}_{j\in I_i}$ is still indistinguishable with $\{<u_{i,j},v_{i,j}>\}_{j\in I_i}$. Due to the GCAN assumption and the entropy assumption, we know each $\{<u_{i,j},v_{i,j}>\}_{j\in I_i}$ and $(u_{k+1},v_{k+1})$ are independent with each other, and $(u_{k+1},v_{k+1})$ is indistinguishable with $(u'_{k+1},v'_{k+1})$. Therefore, $\mathcal{S}$ simulates algorithm 1's output for ($(\vec{v},W)$) input.
\end{proof}

\subsubsection{Analysis of the Divide-and-Conquer Scheme}
The divide-and-conquer scheme break the matching of one rule into several $k$ parts. It is easy to see it has the same amount of security except it reveals the matching results of all parts in addition.

\section{Experiments}
\label{sec6}

To evaluate the feasibility and performance of SOFA, we implement our algorithms in C++, and carry out experiments on real-world firewall rules. The cryptographic multilinear map we employ is the one constructed over integers\cite{coron2013practical, mmapcode}.

\subsection{Experiment Setup}
 Our experiments are conducted on a PC running Linux  with 2.90 GHz  \textit{Inter(R) Core(TM) i7} CPU and 8 GB RAM. And we use 50 standard real-world firewall rules to evaluate the scalability of our algorithms. Note that our blocking scheme can also deal with the extended rules with minor modification.  To leave out unnecessary details, these rules are preprocessed such that  only the fields in the IP/TCP header are included in $r$. In reality, many firewall rules belong to this type.

The generation of an instance of a practical multilinear map involves  the configuration of many parameters\cite{coron2013practical}. And the configuration will  affect its performance and  usability, which is related to the control of the ``noise". So it should be tuned with specific application. Thus we do not present the tedious detailed parameters of our experiments here. Instead our settings ensure the ``noise" in our experiments is under control such that no correctness will be sacrificed for efficiency.

\begin{figure}[!t]

\centering
\includegraphics[height=1.7in,width=2.5in]{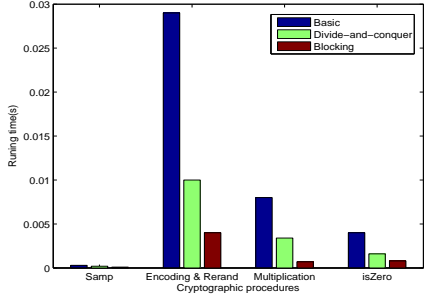}
 %where an .eps filename suffix will be assumed under latex,
% and a .pdf suffix will be assumed for pdflatex; or what has been declared
% via \DeclareGraphicsExtensions.
\caption{ Computation overheads of primitive procedures in GES}
\label{primitive}
\end{figure}

\subsection{Primitive Cryptographic Overheads}
The procedures of GES that we have defined are the basic steps in our algorithms, so we first measure the efficiency of these primitive cryptographic operations. We set the multilinearity level $\kappa$=5, 9 and 33, which are the parameters for obfuscating a standard rule each for the basic scheme, the divide-and-conquer scheme and the one with blocking approach. The results are shown in figure \ref{primitive}.

We can observe from the results that the \textit{ring sampler} takes only less than 1 millisecond. The \textit{encoding and re-randomization} procedure are always performed together, so we measure them together, and we can observe that the multilinearity level affects its running time greatly.
%These procedures are performed in the first offline phase.
The \textit{ multiplication} procedure is the major operation in the second phase, and it can be seen that its efficiency is also strongly affected by the multilinearity level, which determines the specific system parameters. We can also see that the \textit{zero testing} procedure is relatively lightweight in computation overheads.

\subsection{Performance in Obfuscation Phase}
The first phase, the obfuscation of original firewall rules, is performed locally by the enterprise SDN controller. We first measure the impact of the bit numbers $n$ (note that the corresponding multilinearity level in GES is $n+1$) on the performance of naive scheme (Fig.\ref{fig_first_case}). The time overheads of procedure \textit{Instance Generation} and  the obfuscating time of a single standard rule are measured separately. From Fig.\ref{fig_first_case}, we can see that both of them grow linearly with respect to the bit size. Specifically, when $n=32$ ( which we need when naive scheme is adopted to protect the standard rules), the first takes about 0.8 seconds, while the obfuscating for a rule takes about 4.1 seconds. This shows the inefficiency of naive scheme.
\begin{figure}[!t]
\centering
\subfloat[Naive Scheme with Various Bits]{\includegraphics[width=1.7in]{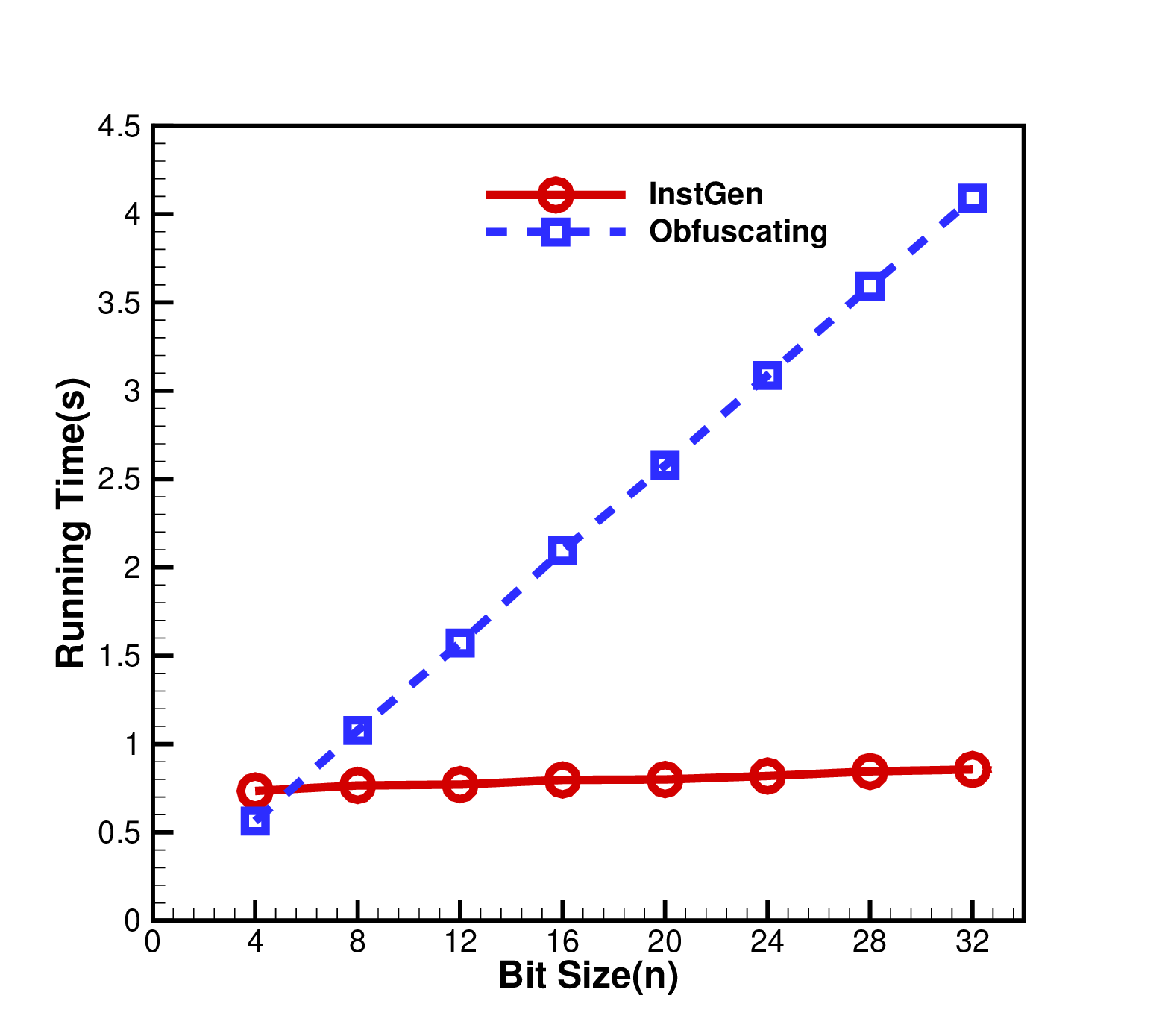}%
\label{fig_first_case}}
\hfil
\subfloat[ Naive and Basic scheme]{\includegraphics[width=1.7in]{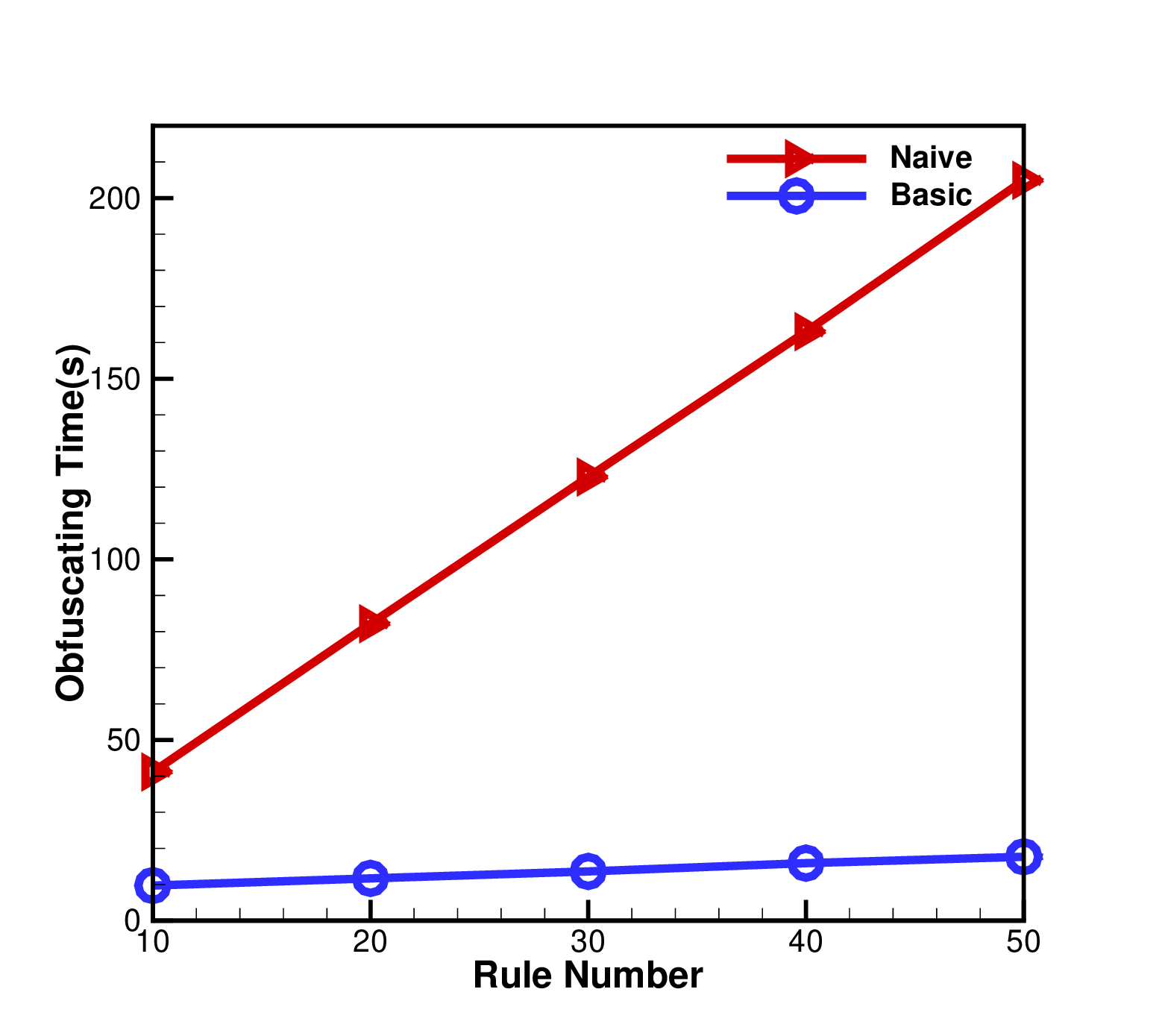}%
\label{fig_second_case1}}
\hfil
\subfloat[ Divide-and-Conquer Scheme]{\includegraphics[width=1.7in]{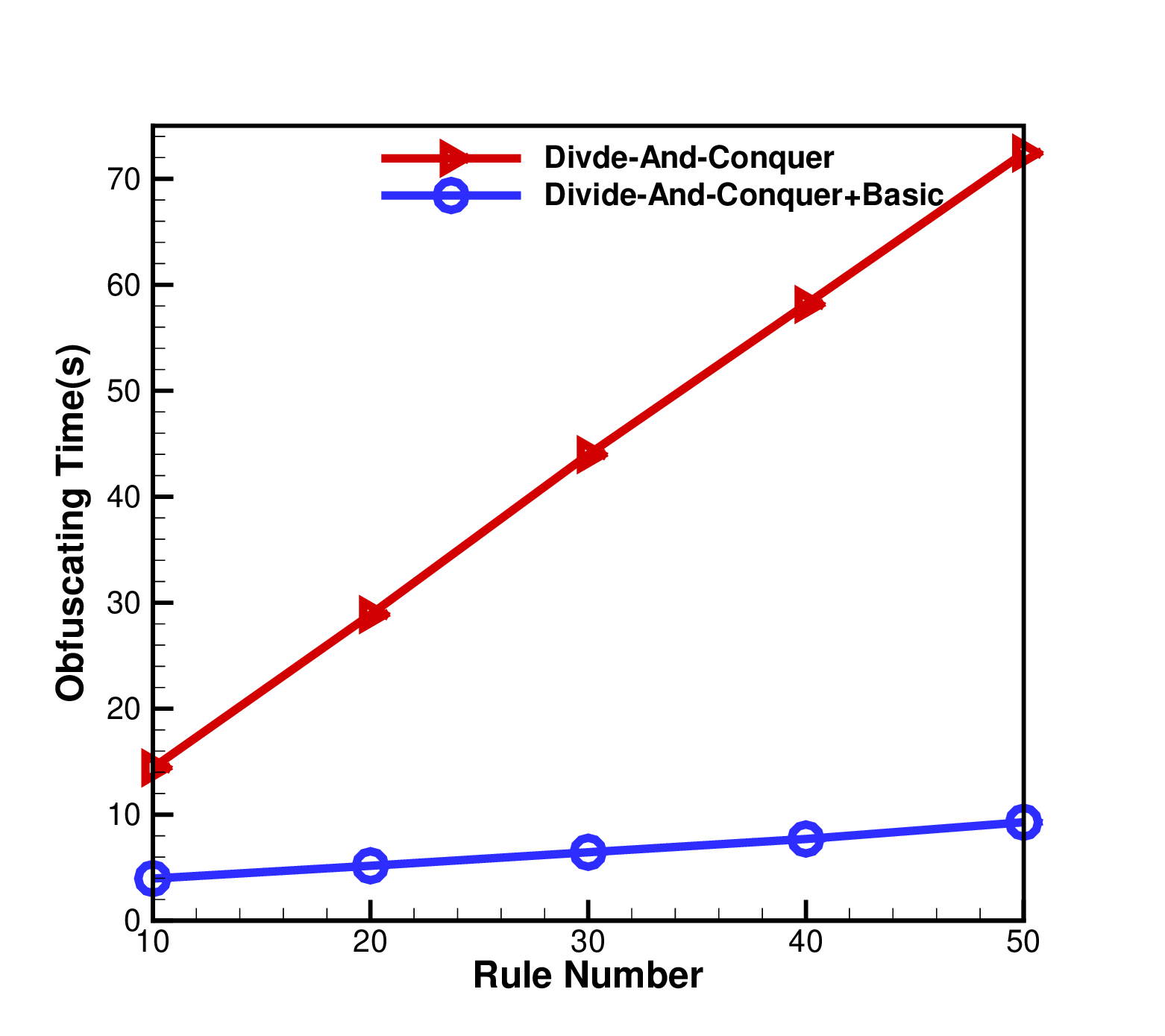}%
\label{fig_second_case2}}
\hfil
\subfloat[ Blocking Scheme]{\includegraphics[width=1.7in]{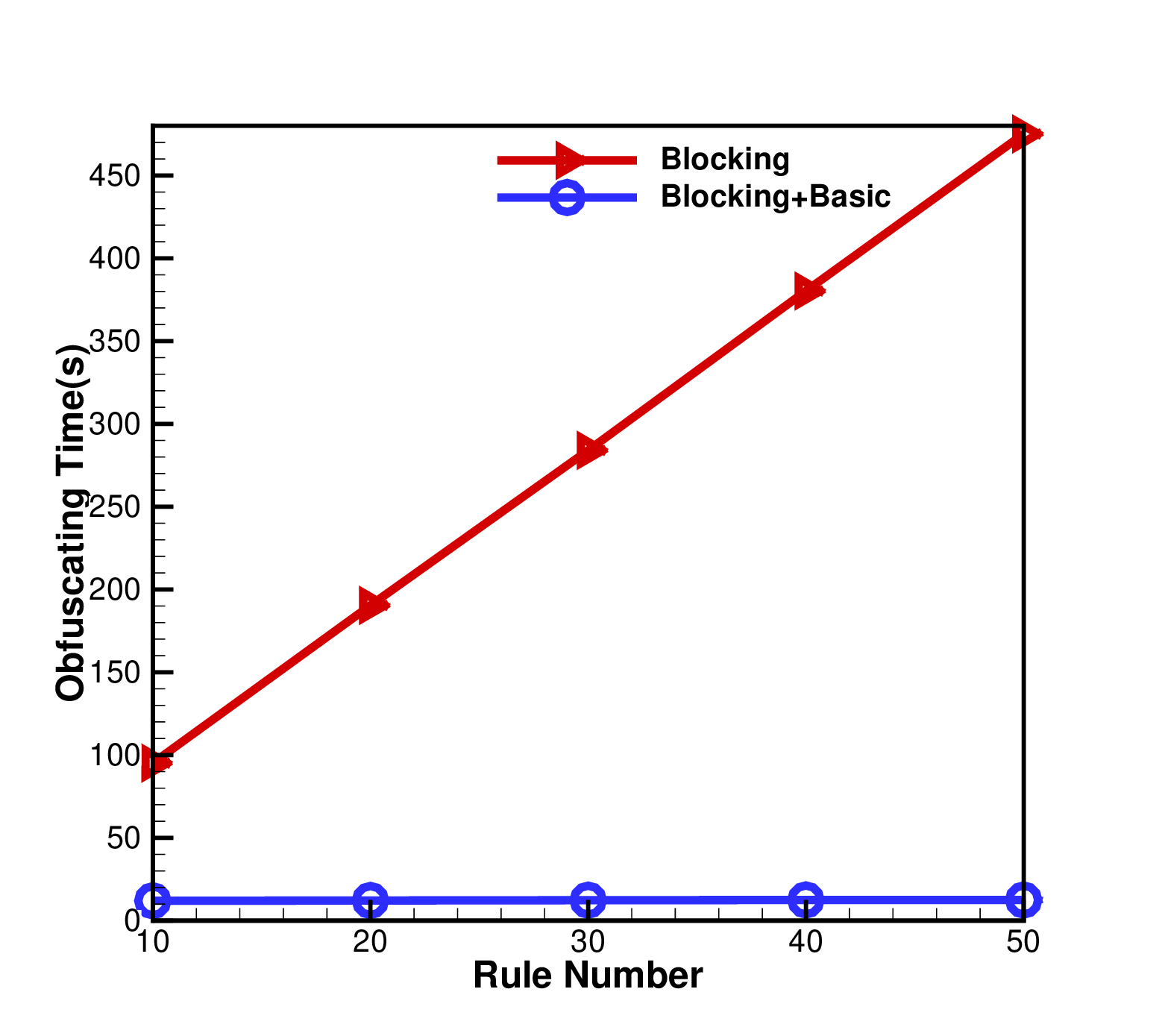}%
\label{fig_second_case3}}
\caption{ Time Performance in the First Phase}
\label{obfuscation}
\end{figure}

\begin{figure}[!t]

\centering
\includegraphics[height=2.0in,width=2.5in]{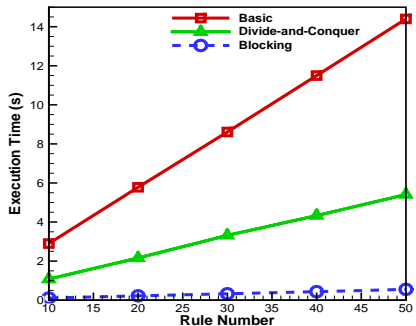}
 %where an .eps filename suffix will be assumed under latex,
 %and a .pdf suffix will be assumed for pdflatex; or what has been declared
 %via \DeclareGraphicsExtensions.
\caption{Time Performance in the Second Phase}
\label{execution}
\end{figure}

 And we also measure the scalability of the schemes we proposed (Fig.\ref{fig_second_case1}, Fig.\ref{fig_second_case2}, Fig.\ref{fig_second_case3}) and the effect of basic scheme on other approaches. The results show that by the encodings compression of our basic scheme, the obfuscation of original rules can be greatly accelerated. This is in line with our theoretical analysis. These results also show that our framework is elastic with different number of rules.

\subsection{Performance in Execution Phase}
The performance in the second phase, the real-time execution of the obfuscated firewalls in the cloud, is more crucial. We measure our solution's performance in this phase with different number of rules, the result is shown in Fig.\ref{execution}.
From the result we can see that the blocking scheme is much more efficient. Actually, we find that the running time to check a packet with them is separately 300, 100, 6 milliseconds. Considering that our experiments are performed on a ordinary PC and no modern firewall process acceleration techniques (such as FDD) are adopted, our schemes, especially blocking scheme, are practical. And in practice, when adopted in cloud processing, many parallelism and acceleration means can be taken to make our scheme faster.

\section{Conclusion}
\label{sec7}
In this paper, we propose solutions to make it secure for an enterprise to outsource firewalls. The security of our schemes have solid theoretical foundation.
 %the cryptographic multilinear map and program obfuscation model.
And the experiment results show that our schemes are also practical to use.

There are more other network functionalities, such as IDS, IPS, that can be outsourced, and their privacy issues must be considered too when outsourcing. We think the framework we propose also works for them because many of them  are also based on packet header filtering. But since their configurations are different from firewalls, so the specific solutions to them will also vary. We leave this as our future work.
%We also believe that the cryptographic multilinear map technology can be useful in other network security and privacy issues.

\bibliographystyle{IEEEtran}
\bibliography{IEEEabrv,Mo}

\end{document}